\documentclass[aps,pra,superscriptaddress,notitlepage,twocolumn,showkeys]{revtex4-1}
\usepackage[utf8]{inputenc}
\usepackage[english]{babel}
\usepackage[T1]{fontenc}
\usepackage{amsmath}
\usepackage{hyperref}

\usepackage{tikz}
\usepackage{lipsum}
\usepackage{amsmath,amssymb,array}
\usepackage{siunitx}
\usepackage{longtable,tabularx}
\usepackage[ruled,lined]{algorithm2e}
\usepackage{algorithmic}
\usepackage{multirow}
\usepackage{color}
\usepackage{bm}
\usepackage{subfigure}
\usepackage{float}

\DeclareMathOperator{\Tr}{\mathrm{Tr}}

\newtheorem{theorem}{Theorem}

\newtheorem{lemma}[theorem]{Lemma}%
\newtheorem{remark}{Remark}%

\newcommand{\argmin}{\mathop{\rm argmin}\limits}

\def\QED{\mbox{\rule[0pt]{1.5ex}{1.5ex}}}

\newcommand{\qed}{\hfill \QED}

\newcommand{\nc}{\newcommand}
\nc{\ketbra}[2]{|#1\rangle\!\langle#2|}
\nc{\bra}[1]{\langle#1|}
\nc{\ket}[1]{|#1\rangle}

\def\Label#1{\label{#1}\ [\ \text{#1}\ ]\ }
\def\Label{\label}
\newcommand{\gl}[1]{{#1}}
\begin{document}

\title{Generalized quantum Arimoto-Blahut algorithm and its application to quantum information bottleneck}

\author{Masahito~Hayashi}
\email{hmasahito@cuhk.edu.cn, masahito@math.nagoya-u.ac.jp}
\affiliation{School of Data Science, The Chinese University
of Hong Kong, Shenzhen, 518172, China}
\affiliation{International Quantum Academy (SIQA), Futian District, Shenzhen 518048, China}
% \affiliation{Guangdong Provincial Key Laboratory of Quantum Science and Engineering, Southern University of Science and Technology, Shenzhen,
% 518055, China}
\affiliation{Graduate School of Mathematics, Nagoya University, Nagoya, 464-8602, Japan}
% \homepage{http://quantum-journal.org}
% \orcid{0000-0002-2445-2701}
\author{Geng Liu}
\affiliation{School of Science and Engineering, The Chinese University
of Hong Kong, Shenzhen, 518172, China}
\affiliation{International Quantum Academy (SIQA), Futian District, Shenzhen 518048, China}

\begin{abstract}
We generalize the quantum Arimoto-Blahut algorithm by Ramakrishnan et al. (IEEE Trans. IT, 67, 946 (2021)) to a function defined
over a set of density matrices with linear constraints so that
our algorithm can be applied to optimizations of quantum operations. 
This algorithm has wider applicability, and
we apply our algorithm to the quantum
information bottleneck with three quantum systems,
which can be used for quantum learning.
We numerically compare our obtained algorithm with the existing algorithm by 
Grimsmo and Still (Phys. Rev. A, 94, 012338 (2016)).
Our numerical analysis shows that our algorithm is better than their algorithm.
\end{abstract}


\maketitle

\section*{Introduction}\Label{setup}
Recently, the technologies for operating quantum systems have been developed quite rapidly.
In order to employ these technologies for a certain purpose in a more efficient way,
we need to choose an optimal operation.
That is, optimizing quantum operation is a crucial task in quantum information processing.
This problem is formulated as an optimization of a certain objective function defined for our quantum operation.
Thus, this kind of optimization is strongly desired.

As an optimization for quantum information,
quantum Arimoto-Blahut algorithm is known very well \cite{Nagaoka}.
This algorithm is the quantum version of Arimoto-Blahut algorithm \cite{Blahut,Arimoto},
and enables the optimization of the mutual information.
Recently, this algorithm has been generalized in a general form \cite{RISB}.
This algorithm is an iterative algorithm,
and converges to the global optimum under a certain condition.
However, it is given as an optimization of a function of 
input density matrices.
Hence, it cannot be applied to the optimization of quantum operations.
The first contribution of this paper is 
to generalize their algorithm to the way that contains the 
case when the objective function has a quantum operation as the input.
This algorithm can be regarded as a quantum extension of the 
algorithm obtained in \cite{Hmixture}
because our algorithm can be considered as an optimization 
over density matrices with linear constraints.
In addition, we show that 
the obtained iterative algorithm improves 
the value of the objective function. 
Since our algorithm can be applied to a function with a very general form,
it can be expected that the algorithm has wide applicability.

As a suitable example of our optimization problem,
we focus on information bottleneck, which
is known as a method to 
extract meaningful information from the original information source \cite{Tishby}.
Recently, this method has been extended to 
the quantum setting, where
the problem is given as an optimization of 
a certain function of a quantum operation \cite{Grimsmo}.
This topic has been studied in various recent papers \cite{BPP,DHW,HW,SCKW,HY}.
The method by \cite{Grimsmo} is given as an iterative algorithm 
for optimizing the given objective function.
Quantum information bottleneck is formulated 
as follows when a joint state $\rho_{XY}$ on the joint system
${\cal X}\otimes {\cal Y}$.
We choose our quantum operation ${\cal E}$ from ${\cal X}$
to our quantum memory ${\cal T}$.
Then, the value $I(T;X)-\beta I(T;Y)$ is given as a function of 
${\cal E}$, and is set to be our minimizing objective function as Fig.
\ref{bottle}.
Their iterative algorithm is given by using a necessary condition 
of the optimum value based on the derivative of the objective function.
That is, their algorithm iteratively changes our operation
to approximately satisfy the necessary condition.
However, their iterative behavior was not sufficiently studied \cite{Grimsmo} 
while it is known that their iterative algorithm
has a good behavior in special cases including a certain quantum setting as follows.

\begin{figure}[htpb]
    \centering
    \includegraphics[width = \columnwidth]{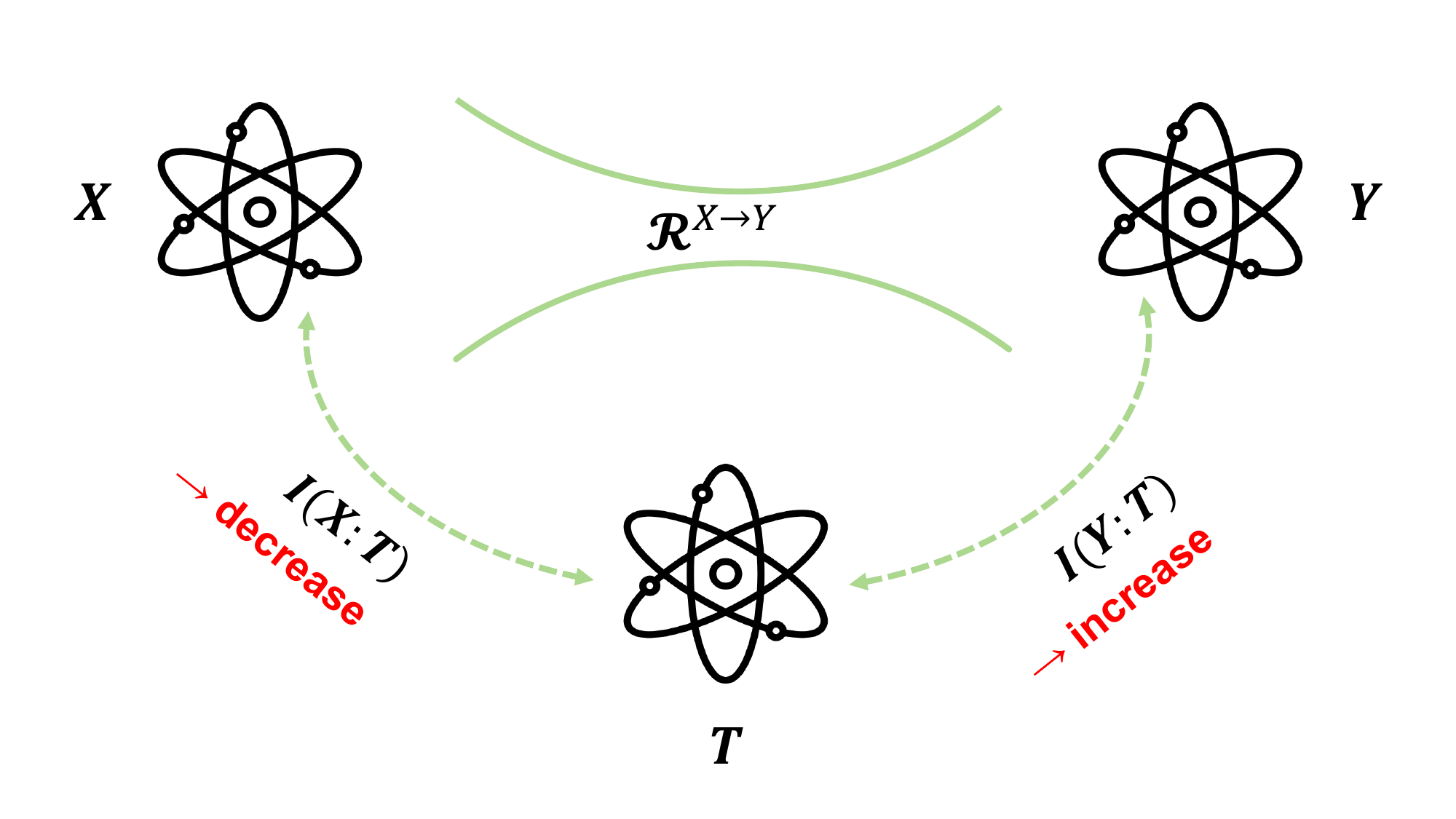}
    \caption{Idea for information bottleneck.}
    \label{bottle}
\end{figure}

First, their iterative algorithm in the classical case coincides with the classical algorithm proposed by \cite{Tishby}.
Second, the recent paper \cite{HY} applied a simple extension of the algorithm by \cite{RISB} to the case when ${\cal X}$ is the classical.
The obtained algorithm by \cite{HY} coincides with the algorithm by \cite{Grimsmo} in this special case.
That is, the algorithm by \cite{HY} coincides with the algorithm by \cite{Tishby} in the classical case.
Also, the paper \cite{HY} showed that their obtained algorithm improves the value of the objective function in each step.
Therefore, it was open whether the algorithm by \cite{Grimsmo} improves the value of the objective function in each step.

As the second contribution of our paper, 
we apply our obtained generalized quantum Arimoto-Blahut algorithm
to quantum information bottleneck even in the case with 
quantum ${\cal X}$.
Notice that the algorithm by \cite{RISB} cannot be applied in this case.
Then, we obtain our concrete iterative algorithm for this minimization problem.
Our obtained algorithm improves the value of the objective function in each step.
However, our obtained iterative step is different from 
the obtained iterative step by \cite{Grimsmo}.
Therefore, it is not clear whether these algorithms are the same.
To clarify this relation, we apply these algorithms to 
the common model of quantum information bottleneck
presented in \cite[Section V-C]{Grimsmo}.
We make their numerical comparisons under this model.
Then, we find that their behaviors are different.
In many cases, our convergent is better than their convergent.
In addition, when we set the initial point to our convergent,
our algorithm does not change the value of the objective function,
but their algorithm increases the value of the objective function a little.
This fact shows that there exists an initial point, where 
their algorithm does not improve the value of the objective function in each step
while it is guaranteed that our algorithm always iteratively improves the value of the objective function.
This fact shows an advantage of our algorithm over their algorithm.

Indeed, the recent papers \cite{HSF1,HSF2,OHS} proposed
minimization algorithms for a function 
of density matrices under linear constraints
by using the mirror descent method.
Since their analyses are based on several conditions of the objective function,
it is not clear whether their method can be applied to the information bottleneck
because it is not clear whether the objective function of the information bottleneck 
satisfies the condition in their methods.

\if0
The remainder of this paper is organized as follows.
Section \ref{S2-1} introduces our generalization of the general algorithm by 
\cite{RISB}.
Section \ref{S5} applies it to quantum information bottleneck.
Section \ref{S6} numerically compares our algorithm for 
quantum information bottleneck with the existing algorithm by \cite{Grimsmo}.
Section \ref{S7} makes discussions and conclusions.
\fi

\section*{Results}
\subsection*{General formulation for minimization problem
over quantum system}
In this paper, we address
a minimization problem over 
a finite-dimensional Hilbert space ${\cal H}$
with a very general form.
Given a continuous function $\Omega$ from ${\cal M}_a$ to 
the set of Hermitian matrices on ${\cal H}$, $B({\cal H})$, 
we consider the minimization $\min_{\rho} {\cal G}(\rho)$ with
\begin{align}
{\cal G}(\rho):= \Tr \rho \Omega[\rho].
\label{RBE1}
\end{align}
For example, the function $\Omega$ can be given by a combination of 
$\log \rho$, linear functions for $\rho$, and constants.
A concrete example will be given in Section \ref{S5}. 
This minimization problem was studied in the existing paper \cite{RISB}.
In physics, we often consider linear constraints to $\rho$
for the range of the minimization.  
However, their algorithm does not work under general linear constraints.
The aim of this paper is to derive an efficient iterative algorithm for the above minimization
that works even under linear constraints.

Before considering this general problem, we discuss
how this setting covers the optimization of a quantum operation.
Assume that ${\cal H}$ is given as a joint system ${\cal A}\otimes {\cal B}$
and a state $\rho_A$ is fixed on ${\cal A}$.
We focus on a TP-CP map ${\cal E}$ from ${\cal A}$ to ${\cal B}$.
We choose a purification $|\phi_{AA'}\rangle$ on
${\cal A}\otimes {\cal A}'$ of $\rho_A$ such that ${\cal A}'$ is isometric to ${\cal A}$.
We have the state 
\begin{align}
\sigma_{AB}:= 
{\cal E}_{A' \to B}(|\phi_{AA'}\rangle\langle \phi_{AA'}|).\label{VSY}
\end{align}
The state $\sigma_{AB}$ satisfies the condition
\begin{align}
\Tr_B\sigma_{AB}=\rho_A.\label{NNE}
\end{align}
Conversely, when a state $\sigma_{AB}$ satisfies the condition 
\eqref{NNE}, there exists a TP-CP map ${\cal E}$ from ${\cal A}$ to ${\cal B}$ to satisfy the condition \eqref{VSY}.
Therefore, 
given a continuous function $\tilde{\Omega}$ 
from $B({\cal A}\otimes {\cal B})$ to itself,
we have the relation between two minimization problems;
\gl{
\begin{align}
& \min_{\sigma_{AB}: \Tr_B \sigma_{AB}=\rho_A} 
\Tr \sigma_{AB} \tilde{\Omega}[\sigma_{AB}] \nonumber \\
=&
\min_{{\cal E}} \Tr {\cal E}_{A' \to B}(|\phi_{AA'}\rangle\langle \phi_{AA'}|)
\tilde{\Omega}[{\cal E}_{A' \to B}(|\phi_{AA'}\rangle\langle \phi_{AA'}|)]. \label{NKA}
\end{align}
}
That is, the above form of the minimization problem for a quantum operation can be converted to 
the minimization problem on the set 
$\{\sigma_{AB}| \Tr_B \sigma_{AB}=\rho_A\}$
with linear constraints.
%Therefore, we discuss the minimization given in \eqref{RBE}.

Now we consider the case
when the system ${\cal A}$ is a classical system.
Hence, ${\cal A}'$ is also a classical system,
the state $\rho_A$ is given as $\sum_{a\in {\cal A}}
P_A(a) |a\rangle \langle a|$,
and the quantum operation ${\cal E}$ is given as 
a map from the classical system ${\cal A}$ to 
a quantum system ${\cal B}$.
That is, ${\cal E}$ maps $a \in {\cal A}$ to 
$\sigma_a \in {\cal S}({\cal B})$.
Hence, the problem \eqref{NKA}
is simplified as
\gl{
\begin{align}
\min_{\{\sigma_a\}_{a \in {\cal A}}} 
& \Tr 
(\sum_{a\in {\cal A}} P_A(a) |a\rangle \langle a|
\otimes \sigma_a) \tilde{\Omega}[
\sum_{a\in {\cal A}} P_A(a) \nonumber\\
& \cdot |a\rangle \langle a|
\otimes \sigma_a]. \label{NKB}
\end{align}
}
% \begin{align}
% \min_{\{\sigma_a\}_{a \in {\cal A}}} 
% \Tr 
% (\sum_{a\in {\cal A}} P_A(a) |a\rangle \langle a|
% \otimes \sigma_a)
% \tilde{\Omega}[
% \sum_{a\in {\cal A}} P_A(a) |a\rangle \langle a|
% \otimes \sigma_a]. \label{NKB}
% \end{align}
Although the above minimization is a special case of the minimization with linear constraints,
the preceding paper \cite{HY} essentially addressed 
this special case by a simple extension of 
the algorithm by \cite{RISB}.
Therefore, the aim of this paper is to extend their algorithm \cite{RISB} to 
the general case with general linear constraints.

\subsection*{Generalized quantum Arimoto-Blahut algorithm}\Label{S2-1}
To handle linear constraints, we introduce 
the notation of quantum information geometry \cite{Amari-Nagaoka}.
That is, we employ the concept of
mixture family \cite{Amari-Nagaoka,H23}.
Using $k$ linearly independent
Hermitian matrices $H_1, \ldots, H_k$ on ${\cal H}$
and constants $a=(a_1, \ldots, a_k)$, 
we define the mixture family ${\cal M}_a$ as follows
\begin{align}
{\cal M}_a:= \{\rho \in {\cal S}({\cal H})| \Tr \rho H_i=a_i \hbox{ for } i=1, \ldots, k
\},\Label{MDP}
\end{align}
where ${\cal S}({\cal H})$ is
the set of density matrices on ${\cal H}$.
We assume that $\dim {\cal H}=d$.
Setting $l:=d^2-1 $,
we added additional $l-k$ linearly independent Hermitian matrices 
$H_{k+1}, \ldots H_l$ such that
$H_{1}, \ldots H_l$ are linearly independent.
Then, a density matrix $\rho$ can be parameterized by 
the mixture parameter $\eta=(\eta_1, \ldots, \eta_l)$ as
$ \eta_i= \Tr \rho H_i$.
That is, the above density matrix is denoted by $\rho_\eta$.
Therefore, our problem is formulated as the following two problems;
\begin{align}
\overline{{\cal G}}(a):=\min_{\rho \in {\cal M}_a} {\cal G}(\rho), \quad
\rho_{*,a}:=\argmin_{\rho \in {\cal M}_a} {\cal G}(\rho).
\label{RBE}
\end{align}

\begin{table}[t]
\caption{List of mathematical symbols for general setting}
\label{symbols}
\begin{center}
\begin{tabular}{|l|l|l|}
\hline
Symbol& Description & Eq. number  \\
\hline
$B({\cal H})$ & Set of Hermitian matrices on ${\cal H}$ &\\
\hline
${\cal S}({\cal H})$ & Set of density matrices on ${\cal H}$ & \\
\hline
${\cal G}(\rho)$ & Objective function & \eqref{RBE1} \\
\hline
${\cal M}_a$ & Mixture family & \eqref{MDP} \\
\hline
$\Gamma^{(e)}_{{\cal M}_a}$ & $e$-projection to ${\cal M}_a$
& \eqref{Mix} \\
\hline
\multirow{2}{*}{$\Omega[\rho]$} & Functional of $\rho$ used  & \multirow{2}{*}{\eqref{RBE1}}\\
&for objective function&\\
\hline
$\overline{{\cal G}}(a)$ & Minimum value of ${\cal G}(\rho)$ & \eqref{RBE} \\
\hline
$\rho_{*,a}$ & Minimizer of ${\cal G}(\rho)$ & \eqref{RBE} \\
\hline
${\cal F}_3[\sigma]$ & Functional of $\sigma$ & \eqref{NMU} \\
\hline
$D_\Omega(\rho\|\sigma)$ &
Function of $\rho$ and $\sigma$ related to $\Omega$
&
\eqref{BK1+} \\
\hline
\end{tabular}
\end{center}
\end{table}

To describe our algorithm, 
we introduce the $e$-projection of $\rho$ to ${\cal M}_a$
by 
$\Gamma^{(e)}_{{\cal M}_a}[\rho]$.
That is, $\Gamma^{(e)}_{{\cal M}_a}[\rho]$ is defined as follows \cite{Amari-Nagaoka}.
\begin{align}
\Gamma^{(e)}_{{\cal M}_a}[\rho]:=
\argmin_{ \sigma \in {\cal M}_a}
D(\sigma\|\rho).\Label{Mix}
\end{align}
Using the projection, we have the following equation
for an element of $\sigma \in {\cal M}_a$, which is often called
Pythagorean theorem.
\begin{align}
D(\sigma\|\rho)=
D(\sigma\|\Gamma^{(e)}_{{\cal M}_a}[\rho])
+ D(\Gamma^{(e)}_{{\cal M}_a}[\rho]\|\rho).\Label{NNP}
\end{align}

To calculate the $e$-projection
$\Gamma^{(e)}_{{\cal M}_a}[\rho]$,
we apply Lemma 5 of \cite{H23} to Section III-C of \cite{H23}.
By using the solution $\tau_*=(\tau_*^1, \ldots, \tau_*^k)$ of 
the following equation,
$\Gamma^{(e)}_{{\cal M}_a}[\rho]$ is given as
$C \exp (\log \rho + \sum_{j=1}^k H_j \tau_*^j )$
with a normalizing constant $C$;
\begin{align}
\frac{\partial}{\partial \tau^i}
\log \Tr \exp
(\log \rho + \sum_{j=1}^k H_j \tau^j )
=a_i.
\end{align}
Instead of solving the above equation, we can employ
the following method.
That is, $\tau_*$ is given as
% \gl{
% \begin{align}
% \tau_*:=
% \argmin_{\tau }
% & \log \Big(\Tr \Big(\exp
% (\log \rho + \sum_{j=1}^k H_j \tau^j )\Big)\Big) \nonumber \\
% & -\sum_{i=1}^k \tau^i a_i.
% \end{align}
% }
\begin{align}
\tau_*:=
\argmin_{\tau }
\log \Big(\Tr \Big(\exp
(\log \rho + \sum_{j=1}^k H_j \tau^j )\Big)\Big)
-\sum_{i=1}^k \tau^i a_i.
\end{align}
Notice that the above objective function 
is convex for $\tau$ \cite[Section III-C]{H23}.

For this aim,
we define the density matrix ${\cal F}_3[\sigma] $ as
\begin{align}
{\cal F}_3[\sigma]:= \frac{1}{\kappa[\sigma]}
\exp( \log \sigma -\frac{1}{\gamma} \Omega[\sigma]),\Label{NMU}
\end{align}
where $\kappa[\sigma]$ is the normalization factor
$\Tr \exp( \log \sigma -\frac{1}{\gamma} \Omega[\sigma])$.
%We generalize the algorithm by \cite{RISB} in the way as \cite{Hmixture}.
Then, we propose Algorithm \ref{AL1} as Fig. \ref{algo}.
When the calculation of $ \Omega[\rho]$ and 
the projection is feasible, Algorithm \ref{AL1} is feasible.

\begin{algorithm}
\caption{Minimization of ${\cal G}(\rho)$}
\Label{AL1}
\begin{algorithmic}
\STATE {Choose the initial value $\rho^{(1)} \in \mathcal{M}$;} 
\REPEAT 
%\STATE {Set $\hat{\theta}=\hat{\theta}^{(t)} \in \Theta_{\mathcal{M}}$;} 
\STATE Calculate $\rho^{(t+1)}:=\Gamma^{(e)}_{{\cal M}_a}[{\cal F}_3[\rho^{(t)}]]
$;
\UNTIL{convergence.} 
\end{algorithmic}
\end{algorithm}

\begin{figure}[htpb]
    \centering
    \includegraphics[width = \columnwidth]{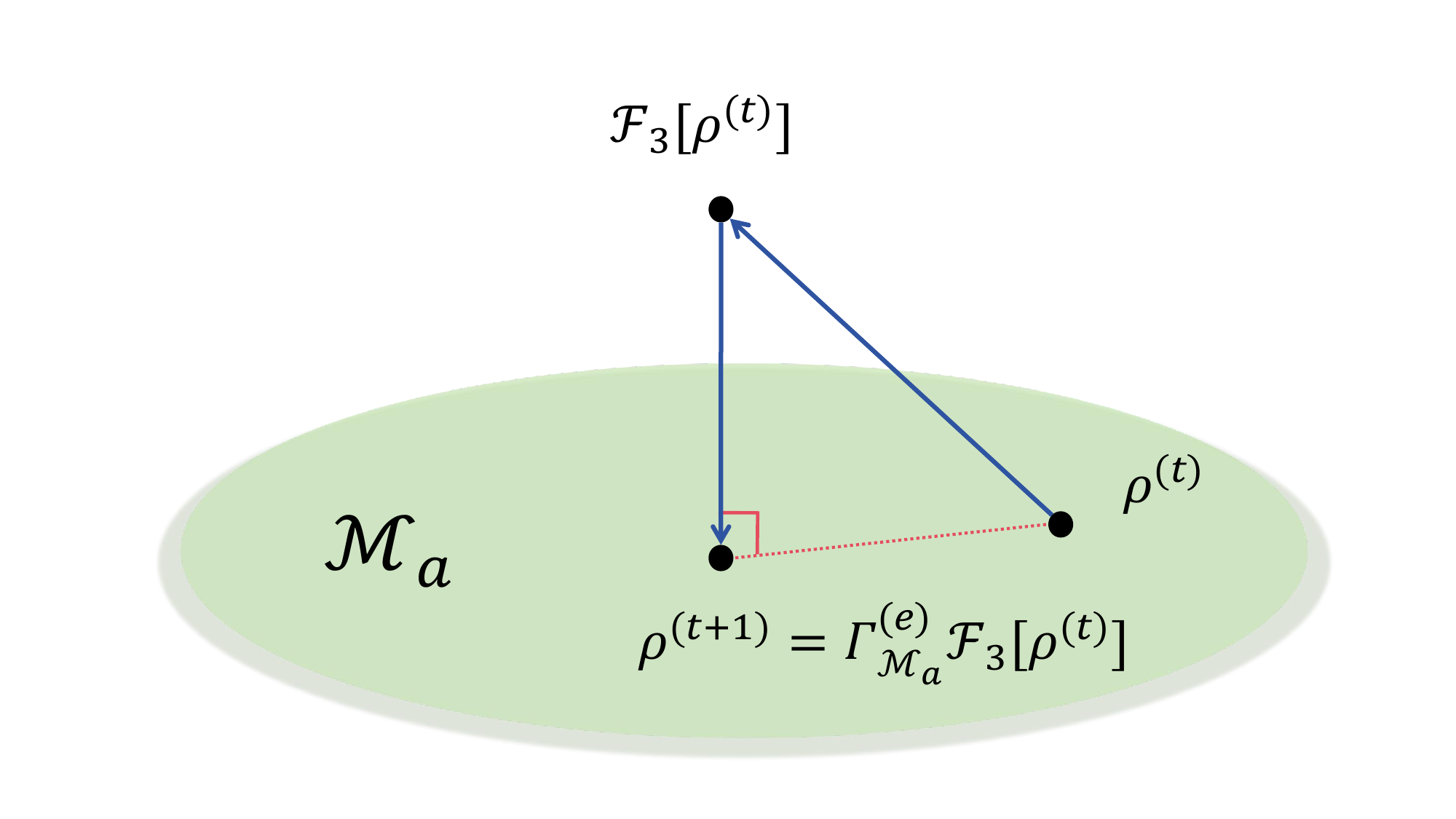}
    \caption{Our generalized algorithm.}
    \label{algo}
\end{figure}

Then, as shown in Method, we have the following two theorems.

\begin{theorem}\Label{NLT}
When all pairs $(\rho^{(t)},\rho^{(t+1)})$ satisfy 
the following condition with $(\rho,\sigma)=(\rho^{(t)},\rho^{(t+1)})$
\begin{align}
D_\Omega(\rho\|\sigma):=\Tr \rho (\Omega[\rho]- \Omega[\sigma])
&\le \gamma D(\rho\|\sigma)
\Label{BK1+} ,
\end{align}
i.e., the positive number $\gamma$ is sufficiently large,
Algorithm \ref{AL1} 
always iteratively improves the value of the objective function. 
%converges to a local minimum.
\end{theorem}

\begin{theorem}\Label{TH1}
When any two densities $\rho$ and $\sigma$ in ${\cal M}_a$ satisfy the condition \eqref{BK1+}, and
the state $\rho_{*,a} $ satisfies 
\begin{align}
D_\Omega(  \rho_{*,a}\|\sigma) \ge 0 
\Label{BK2+} 
\end{align}
with any state $\sigma $,
Algorithm \ref{AL1} satisfies the condition
\begin{align}
{\cal G}(\rho^{(t_0+1)})
-{\cal G}(\rho_{*,a})
\le 
\frac{\gamma D(\rho_{*,a}\| \rho^{(1)}) }{t_0} \Label{XME}
\end{align}
with any initial state $\rho^{(1)}$.
\end{theorem}

The mathematical symbols for our general setting 
is summarized in Table \ref{symbols}

%The proof of Theorem \ref{TH1} is similar to 
%the same as 
%(9) of \cite{RISB}.

\if0
In fact, the condition (A2) is closely related to the convexity of ${\cal G}(\rho)$ as follows.
\begin{lemma}
When 
\begin{align}
\Tr \rho (\Omega[\rho]- \Omega[\sigma]) \ge 0
\Label{XMZ2}
\end{align}
holds for $P,Q \in {\cal M}_a$, 
the map $ P \mapsto {\cal G}(\rho)$ is convex.
\end{lemma}
\begin{proof}
\begin{align}
&\lambda {\cal G}(\rho)+(1-\lambda) {\cal G}(\sigma)
-\lambda {\cal G}(\lambda P+(1-\lambda)\sigma)\\
=&\lambda \Tr \rho (\Omega[\rho]- \Omega[\lambda P+(1-\lambda)Q])\\
&+(1-\lambda) \Tr Q (\Omega[\sigma]- \Omega[\lambda P+(1-\lambda)Q])\\
 \ge & 0.
\end{align}
\end{proof}
\fi

\subsection*{Application to information bottleneck}\Label{S5}
As a method for information-theoretical machine learning,
we often consider quantum information bottleneck 
\cite{Grimsmo}, which is a quantum generalization of 
\cite{Tishby}.
Our approach is the application of Algorithm \ref{AL1}.
Consider a TP-CP map ${\cal R}$ from 
${\cal X}$ to ${\cal Y}$ and a state $\rho_X$ on ${\cal X}$. 
The task is to extract a piece of essential information
from the system ${\cal X}$ to the system ${\cal T}$ with respect to ${\cal Y}$.

The information extraction from the system ${\cal X}$ to the system ${\cal T}$ is given by 
a TP-CP map ${\cal E}$ from
the system ${\cal X}$ to the system ${\cal T}$.
To consider the correlation between 
the systems ${\cal X}$ and ${\cal T}$,
we consider the purification $|\phi_{XX'}\rangle$ of $\rho_X$.
Then, we have the state 
${\cal E}_{X'\to T}(|\phi_{XX'}\rangle\langle\phi_{XX'}|)$
on the joint system ${\cal X}\otimes {\cal T}$.
This joint state is an element of 
${\cal M}(\rho_X):=\{\sigma_{TX}| \Tr_T \sigma_{TX}=\rho_X\}$.
Considering the state ${\cal E}_{X'\to T}(|\phi_{XX'}\rangle\langle\phi_{XX'}|)$, we find a one-to-one relation 
between 
a TP-CP map from
the system ${\cal X}'$ to the system ${\cal T}$
and an element of ${\cal M}(\rho_X)$.
In the following, we discuss 
the information extraction from the system ${\cal X}$ to the system ${\cal T}$ by an element of 
${\cal M}(\rho_X)$.
Let $d$ be the dimension of ${\cal X}$.
We choose $d^2-1 $ linearly independent Hermitian matrices
$H_1, \ldots, H_{d^2-1}$ that are linearly independent of $I$.
The set ${\cal M}(\rho_X)$ is characterized as
\gl{\begin{align}
{\cal M}(\rho_X)=\big\{\sigma_{TX}\big|
& \Tr_{TX} \sigma_{TX} H_j= \Tr_{X}\rho_X H_j \nonumber\\
& \hbox{ for }
j=1, \ldots,d^2-1 \big\}.
\end{align}}
% \begin{align}
% {\cal M}(\rho_X)=\big\{\sigma_{TX}\big|
% \Tr_{TX} \sigma_{TX} H_j= \Tr_{X}\rho_X H_j \hbox{ for }
% j=1, \ldots,d^2-1 \big\}.
% \end{align}
Hence, we find that 
${\cal M}(\rho_X)$ is a mixture family generated by
$H_1, \ldots, H_{d^2-1}$.

For this information extraction, given parameters $\alpha \in [0,1]$ and $\beta \ge \alpha$,
we choose a joint state 
$\sigma_{TX}^*\in {\cal M}(\rho_X)$
as
\gl{
\begin{align}
\sigma_{TX}^*:= & \argmin_{\sigma_{TX} \in {\cal M}(\rho_X)} 
\alpha I(T;X)[\sigma_{TX}]
\nonumber\\
 & +(1-\alpha)H(T)[\sigma_{TX}]
 -\beta I(T;Y)[{\cal R}(\sigma_{TX})].\label{BNU}
\end{align}}
% \begin{align}
% \sigma_{TX}^*:= \argmin_{\sigma_{TX} \in {\cal M}(\rho_X)} 
% \alpha I(T;X)[\sigma_{TX}]+(1-\alpha)H(T)[\sigma_{TX}]
% -\beta I(T;Y)[{\cal R}(\sigma_{TX})].\label{BNU}
% \end{align}
%This formulation is the quantum extension of the 
%formulation that is given by \cite{} as an extension of the original formulation by \cite{}
This method is called information bottleneck.
The above objective function of the classical case is 
generalized from the original 
information bottleneck.
Original information bottleneck \cite{Tishby} has only the parameter $\beta$
and sets $\alpha$ to be $1$ in the classical setting.
The paper \cite{Grimsmo} follows this formulation.
The paper \cite{Strouse} introduced the additional parameter $\alpha$, and 
the paper \cite{HY} follows this extended formulation.

The reference \cite{Grimsmo} proposed an iterative algorithm to solve the above minimum for the case with $\alpha=1$,
it did not discuss the convergence.
The reference \cite{HY}
proved its convergence only 
when ${\cal X}$ is a classical system.
Hence, it is not known what algorithm converges 
when ${\cal X}$ is a quantum system.

To discuss this problem, we apply 
our method to this problem. We define
\begin{align}
{\Omega}_{\alpha,\beta}[\sigma_{TX}]
:=
&\alpha \log \sigma_{TX}
-\alpha \log \rho_{X}
+(\beta-1) \log \sigma_{T} \nonumber \\
&+\beta 
 {\cal R}^* ( \log {\cal R}(\rho_{X}) 
- \log {\cal R}( \sigma_{TX}) ).\Label{NBT}
\end{align}
Then, as shown in Method,
the objective function is written as
\begin{align}
{\cal G}_{\alpha,\beta}( \sigma_{TX}):=
\Tr\sigma_{TX} {\Omega}_{\alpha,\beta}[\sigma_{TX}].\label{NMD}
\end{align}
That is, our problem is reduced to the minimization
\begin{align}
\min_{\sigma_{TX}\in {\cal M}(\rho_X)}
{\cal G}_{\alpha,\beta}(\sigma_{TX}).
\end{align}
Then, the mathematical symbols for
information bottleneck 
is summarized in Table \ref{symbols2}.

Then, we apply Algorithm \ref{AL1} to this problem.
${\cal F}_3[\sigma_{TX}]$ is calculated as
\gl{\begin{align}
&{\cal F}_3[\sigma_{TX}]
={\kappa_1}
\exp( \log \sigma_{TX} -\frac{1}{\gamma} \Omega_{\alpha,\beta}[\sigma_{TX}]) 
\nonumber\\
=&{\kappa_1}
\exp \Big( 
(1-\frac{\alpha}{\gamma}) \log \sigma_{TX}
+\frac{\alpha}{\gamma}  \log \rho_{X}
\nonumber\\
&-\frac{\beta-1 }{\gamma} \log \sigma_{T} 
-\frac{\beta}{\gamma}  
 {\cal R}^* ( \log {\cal R}(\rho_{X}) 
- \log {\cal R}( \sigma_{TX}) \Big),
\end{align}}
% \begin{align}
% &{\cal F}_3[\sigma_{TX}]
% ={\kappa_1}
% \exp( \log \sigma_{TX} -\frac{1}{\gamma} \Omega_{\alpha,\beta}[\sigma_{TX}]) 
% \nonumber\\
% =&{\kappa_1}
% \exp \Big( 
% (1-\frac{\alpha}{\gamma}) \log \sigma_{TX}
% +\frac{\alpha}{\gamma}  \log \rho_{X}
% -\frac{\beta-1 }{\gamma} \log \sigma_{T} 
% -\frac{\beta}{\gamma}  
%  {\cal R}^* ( \log {\cal R}(\rho_{X}) 
% - \log {\cal R}( \sigma_{TX}) \Big),
% \end{align}
where $\kappa_1$ is the normalizing constant.
Hence, by using
the $e$-projection $\Gamma^{(e)}_{{\cal M}(\rho_X)}$
to ${\cal M}(\rho_X)$,
the update rule is given as
$\sigma_{TX}\in {\cal M}(\rho_X)
\mapsto 
\Gamma^{(e)}_{{\cal M}(\rho_X)}
[{\cal F}_3[\sigma_{TX}]]$.

\begin{table}[t]
\caption{List of mathematical symbols for information bottleneck}
\label{symbols2}
\begin{center}
\begin{tabular}{|l|l|l|}
\hline
Symbol& Description & Eq. number  \\
\hline
\multirow{2}{*}{${\cal M}(\rho_X)$} & Mixture family of states  & 
\\
&whose reduced density is $\rho_X$&\\
\hline
${\cal G}_{\alpha,\beta}(\rho)$ & Objective function & \eqref{NMD} \\
\hline
\multirow{2}{*}{${\Omega}_{\alpha,\beta}[\sigma_{TX}]$}
&Functional of $\sigma_{TX}$ used&\multirow{2}{*}{\eqref{NBT}}\\
&for objective function&\\
\hline
\end{tabular}
\end{center}
\end{table}

To get this, we need to find $\tau_{*}=(\tau_{*}^1, \ldots, \tau_{*}^{d^2-1})$
satisfying the following condition.
\begin{align}
& \frac{\partial}{\partial \tau^i}
\log \Tr \exp
\Big(\log \sigma_{TX} -\frac{1}{\gamma} \Omega_{\alpha,\beta}[\sigma_{TX}] 
+ \sum_{j=1}^{d^2-1} H_j \tau^j \Big)
\nonumber\\
&=\Tr \rho_X H_i
\end{align}
% \begin{align}
% \frac{\partial}{\partial \tau^i}
% \log \Tr \exp
% \Big(\log \sigma_{TX} -\frac{1}{\gamma} \Omega_{\alpha,\beta}[\sigma_{TX}] 
% + \sum_{j=1}^{d^2-1} H_j \tau^j \Big)
% =\Tr \rho_X H_i
% \end{align}
for $i=1, \ldots, d^2-1$.
Then, 
$\Gamma^{(e)}_{{\cal M}(\rho_X)}
[{\cal F}_3[\sigma_{TX}]]$ is given as
\gl{\begin{align}
&\Gamma^{(e)}_{{\cal M}(\rho_X)}
[{\cal F}_3[\sigma_{TX}]]
\nonumber\\
=& \kappa_2
\exp
(\log \sigma_{TX} -\frac{1}{\gamma} \Omega_{\alpha,\beta}[\sigma_{TX}] 
+ \sum_{j=1}^k H_j \tau_*^j ),
\end{align}}
% \begin{align}
% \Gamma^{(e)}_{{\cal M}(\rho_X)}
% [{\cal F}_3[\sigma_{TX}]]
% =\kappa_2
% \exp
% (\log \sigma_{TX} -\frac{1}{\gamma} \Omega_{\alpha,\beta}[\sigma_{TX}] 
% + \sum_{j=1}^k H_j \tau_*^j ),
% \end{align}
where $\kappa_2$ is the normalizing constant.

Also, as shown in Method,
%as shown in \cite[(22)]{HY},
we have the relation 
\begin{align}
D_{\Omega_{\alpha,\beta}}( \sigma_{TX} \|\sigma_{TX}')
\le \alpha D( \sigma_{TX} \|\sigma_{TX}')\Label{MMS}
\end{align}
for $\sigma_{TX},\sigma_{TX}' \in {\cal M}(\rho_X)$.
That is,
the condition \eqref{BK1+} holds with $\gamma=\alpha$.
Due to Theorem \ref{NLT},
our algorithm with $\gamma=\alpha$
converges.

From the construction,
when the system $X$ is a classical system.
Our algorithm is the same as the algorithm by \cite{HY}.
Also, it was shown in \cite{HY} that 
the algorithm by \cite{HY} with $\gamma=\alpha$ is the same as 
the algorithm given in \cite{Grimsmo} with a classical system $X$.
However, the above analysis does not clarify whether the above algorithm is the same as 
the algorithm given in 
\cite{Grimsmo} 
when the system $X$ is not a classical system.

%\section{Numerical comparison}\Label{S6}
\subsection*{Numerical analysis}
\gl{To clarify whether our proposed algorithm is the same as the algorithm given in \cite{Grimsmo} with a general quantum system $X$, a comparative analysis is conducted between the two algorithms. To this end, we focus on the concrete model introduced in \cite[Section V-C]{Grimsmo}, which is based on an amplitude-damping channel. We choose the system $X$ as the combination of two qubits $X_1$ and $X_2$.}
The state $\rho_{X_1X_{2}}$ and
the TP-CP map ${\cal R}$ are given as 
\begin{align}
\rho_{X_1X_{2}}&=[(1-p)\ketbra{+}{+}+p\frac{I_2}{2}]\otimes \frac{I_2}{2} \\
\mathcal{R}(\sigma_{R_1R_{2}}) &= \Tr_{R_2}[K_1\sigma_{R_1R_2}K_1^{\dagger} + K_2\sigma_{R_1R_2}K_2^{\dagger}],
\end{align}
where
\begin{align}
K_1 = \begin{pmatrix}
    1 & 0\\ 0 & \sqrt{1-\lambda}
\end{pmatrix} \otimes I_2 ,\quad
K_2 = \begin{pmatrix}
    0 & \sqrt{\lambda}\\ 0 & 0
\end{pmatrix} \otimes I_2.
\end{align}
In the above channel, an amplitude-damping
channel acts on the first qubit, and the second qubit is traced out.
\gl{The system $T$ is chosen to have only one qubit, which means the dimension of $T$ is $d_T=2$.
We choose the parameters $\lambda = 0.7$ and $p = 0.3$. 
We set $\alpha = 1$ in the objective function in \eqref{BNU}
so that we minimize $I(T;X)[\sigma_{TX}]-\beta I(T;Y)[{\cal R}(\sigma_{TX})]$.
To ensure the convergence of our algorithm, we set $\gamma$ equal to $\alpha=1$ within our algorithmic framework.}

\begin{figure}[htbp]
    \centering
    % \subfigure{\includegraphics[width=  0.6\columnwidth]{Figure/QQQ_crx_0.3.png}}
    % \subfigure{\includegraphics[width = 0.6\columnwidth]{Figure/QQQ_crx_0.1.png}}
    \subfigure{\includegraphics[width=  \columnwidth]{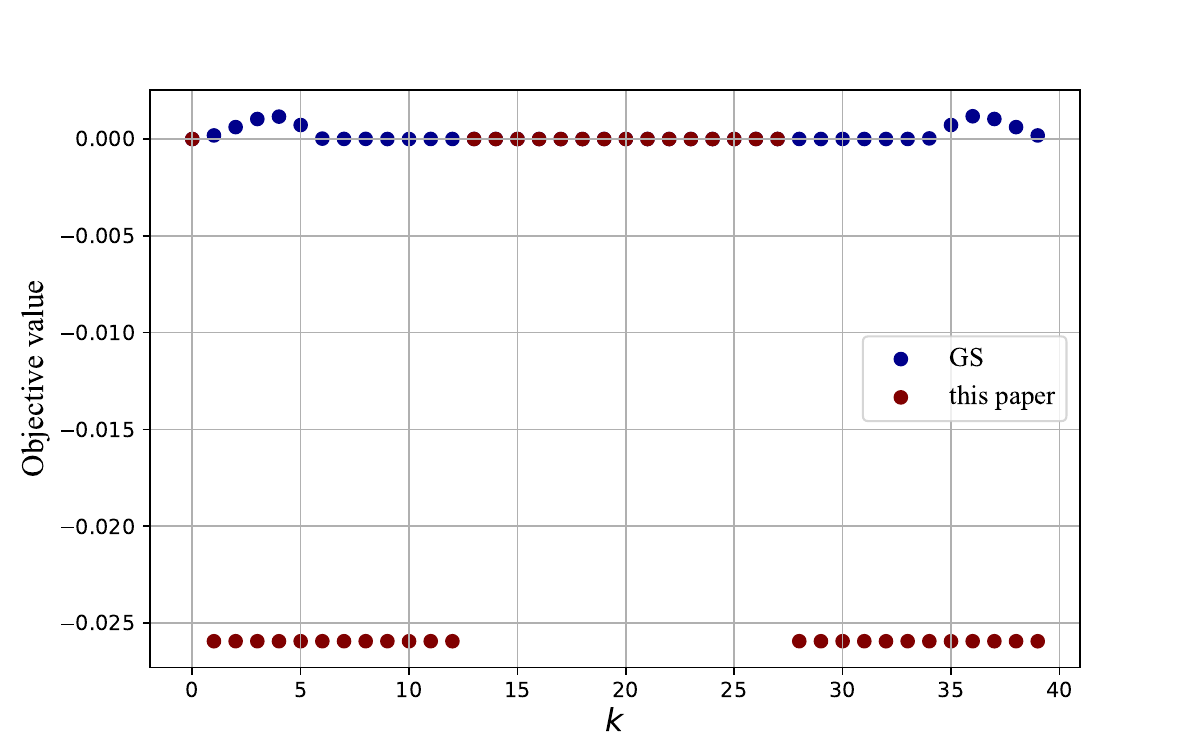}}
    \subfigure{\includegraphics[width=  \columnwidth]{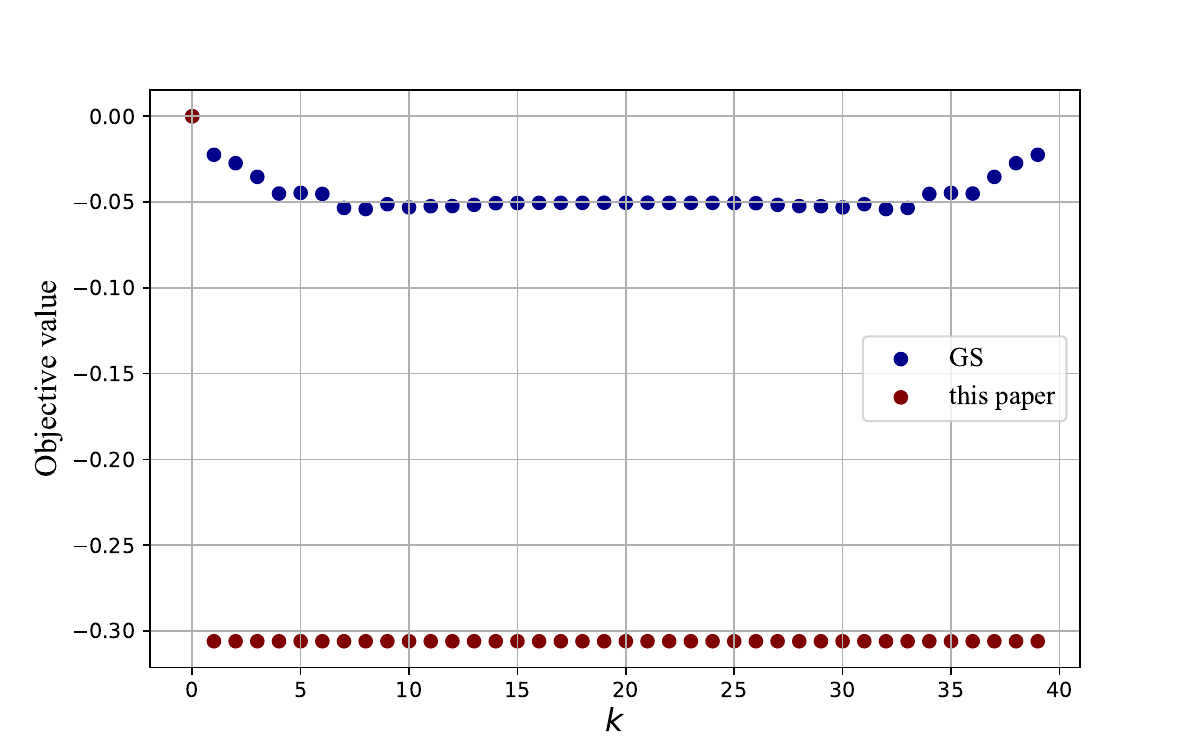}}
    \subfigure{\includegraphics[width = \columnwidth]{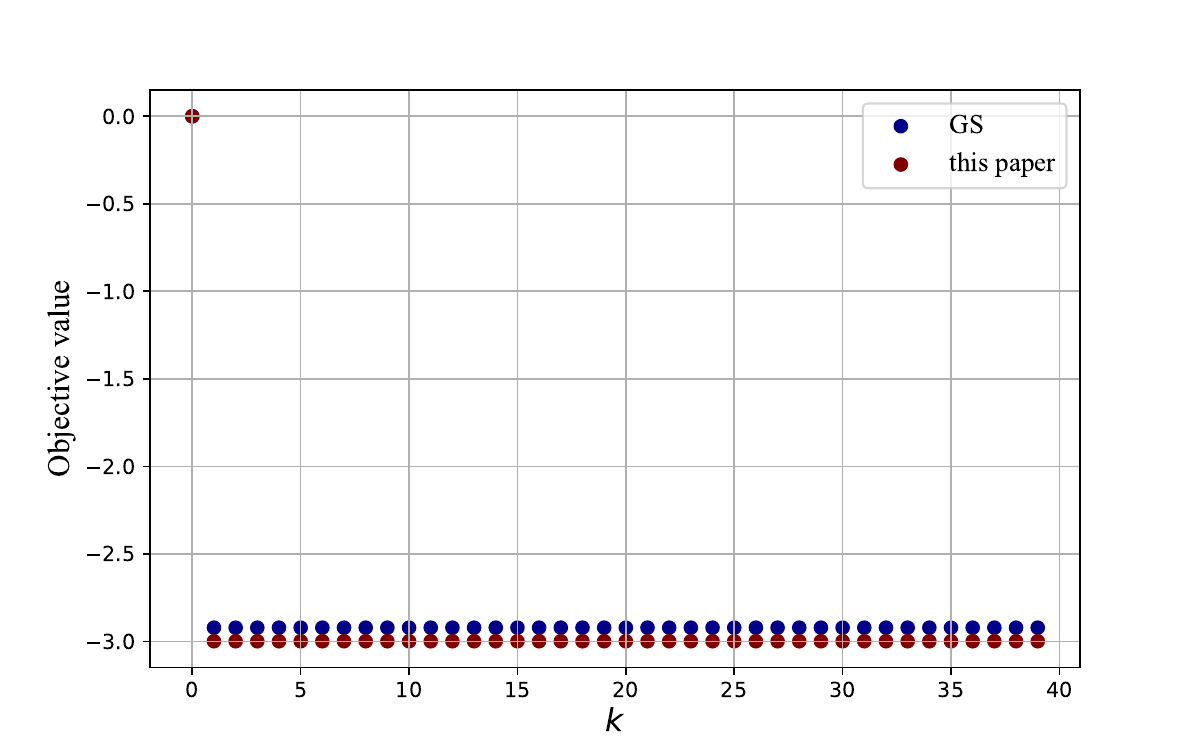}}
    \caption{\gl{The convergent value of $I(T;X)[\sigma_{TX}]-\beta I(T;Y)[{\cal R}(\sigma_{TX})] $ as a function of $k$. 
    The vertical axis expresses
    the convergent value of the objective function.
%    the value of the objective function after 10 iterations.
    The horizontal axis shows the choice of $k$ from $0$ to $39$. The figures from top to bottom are the result of $\beta = 2,10/3, 10$ respectively.
    In the middle figure, these two convergent values take almost the same value from $k=13$ to $k=27$. Also, in all figures, they take almost the same value at $k=0$. 
     }}
    \label{fig:1}
\end{figure}

%\subsection*{Detail procedure of our numerical calculation}
\gl{As both our algorithm, shown in Algorithm \ref{AL1}, and the algorithm proposed by \cite{Grimsmo}
are iterative in nature, the initial state $ \sigma_{TX}$ should be chosen to subject to the constraint
$\Tr_T \sigma_{TX}=\rho_X$.
Given that both algorithms do not guarantee the convergence to the global minimum,
their convergent points potentially depend on the choice of the initial points.
Therefore, the initial point should be chosen carefully.}

In this paper, we consider the following series of initial states;
\begin{equation}
    \sigma_{XT}^{(k)} = \Tr_{X'_1X'_2}[U_{(k)}U_{X'_1X'_2}(\rho_{XX'_1X'_2}\otimes \rho_T)U^{\dagger}_{X'_1X'_2}U^{\dagger}_{(k)}],
    \label{Eq: different_initial_states}
\end{equation}
where $\rho_{XX'_1X'_2}$ is a purification of $\rho_X$
with two additional qubits $X_1'$ and $X_2'$,
and the state $\rho_T$ is an arbitrary state on system $T$. 
Remember that the system $X$ is 2 qubits and the system $T$ is 1 qubit. 
 The unitary $U_{X'_1X'_2}$ is randomly chosen from the $\mathcal{U}(4)$ and is acting on the system $X'_1X'_2$. 
 The $U_{(k)}$ is acting on the system $X'_2T$ and has the following form
\begin{equation}
    U_{(k)} = \ketbra{0}{0}\otimes I + \ketbra{1}{1}\otimes R_x(\frac{2\pi k}{n}), \label{BVT}
\end{equation}
where the $R_x(\theta)$ denotes the unitary representation of $R_x$ gate with $\theta$ rotation angle and
its mathematical expression is
\begin{equation}
    R_x(\theta) = \begin{bmatrix}
        \cos{\frac{\theta}{2}} & -i \sin{\frac{\theta}{2}}\\
        -i\sin{\frac{\theta}{2}}& \cos{\frac{\theta}{2}}
    \end{bmatrix}.
\end{equation}
\gl{In addition, $n$ is a fixed integer (we choose $n=40$).}
% and $k$ is an integer in $\{0,\ldots, n\}$. 
Since $U_{(k)}$ is the identity when $k=0$, 
the state $\rho_{XT}^{(0)}$ has the product form $ \rho_X\otimes\rho_T$. 
Varying the value of $k$, we can change the correlation between $X$ and $T$ 
in the initial point.

In the first experiment, we investigate the impact of the initial point selection on the convergence behavior of both algorithms by varying the parameter $k$ within the range from $0$ to $39$ with $\beta$ across values of $2, 10/3$ and $10$. Here, parameter $k$ serves as the determinant for the selection of the initial point.
Fig.~\ref{fig:1} illustrates that our algorithm exhibits reduced sensitivity to initial points compared to their algorithm.
Also, these figures show that the choice of $k=20$ is the best choice.
In the second experiment, we consistently set $k$ to be $20$. The results presented in Fig.~\ref{fig:2} demonstrate that our algorithm outperforms theirs across the range of $1/\beta$ values from $0.11$ to $0.3$.

\begin{figure}[htpb]
    \centering
    \includegraphics[width = \columnwidth]{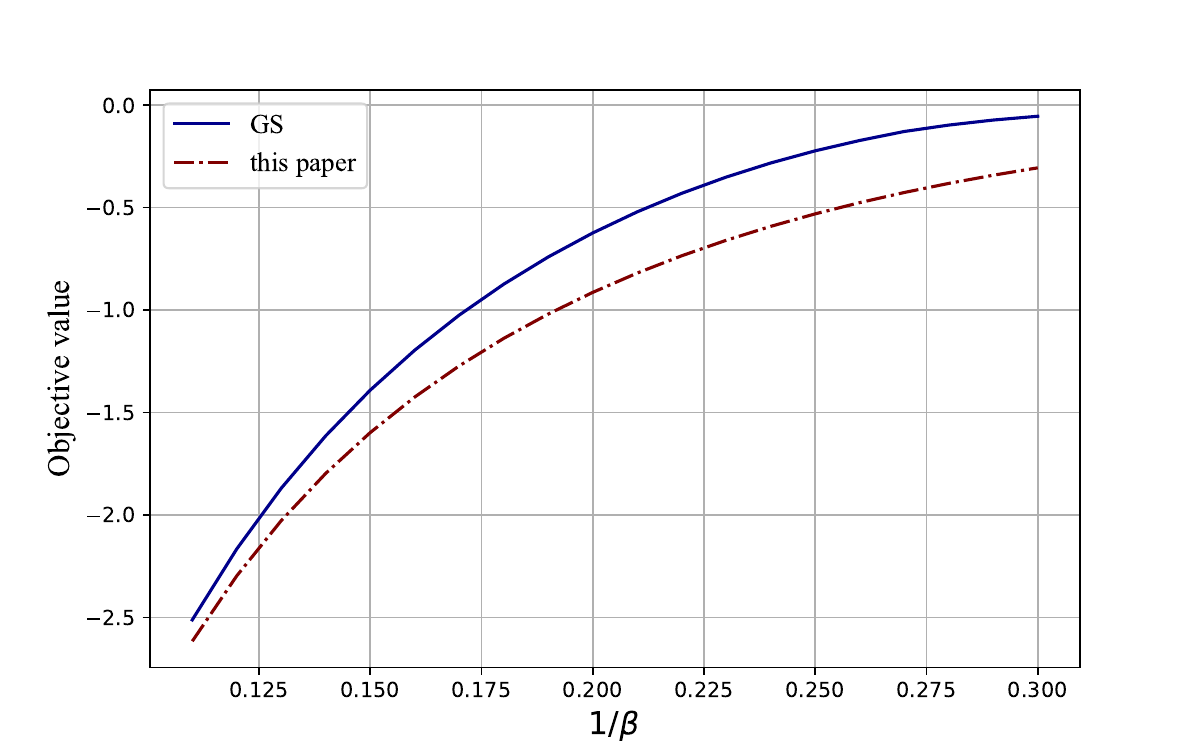}
    \caption{The convergent value of $I(T;X)[\sigma_{TX}]-\beta I(T;Y)[{\cal R}(\sigma_{TX})] $ with $k=20$.
    The vertical axis expresses
    the convergent value of the objective function.
%    the value of the objective function after 10 iterations.
    The horizontal axis shows the choice of $1/\beta$ from $0.11$ to $0.3$.   }
    \label{fig:2}
\end{figure}

\begin{figure}[htbp]
    \centering
    \includegraphics[width = \columnwidth]{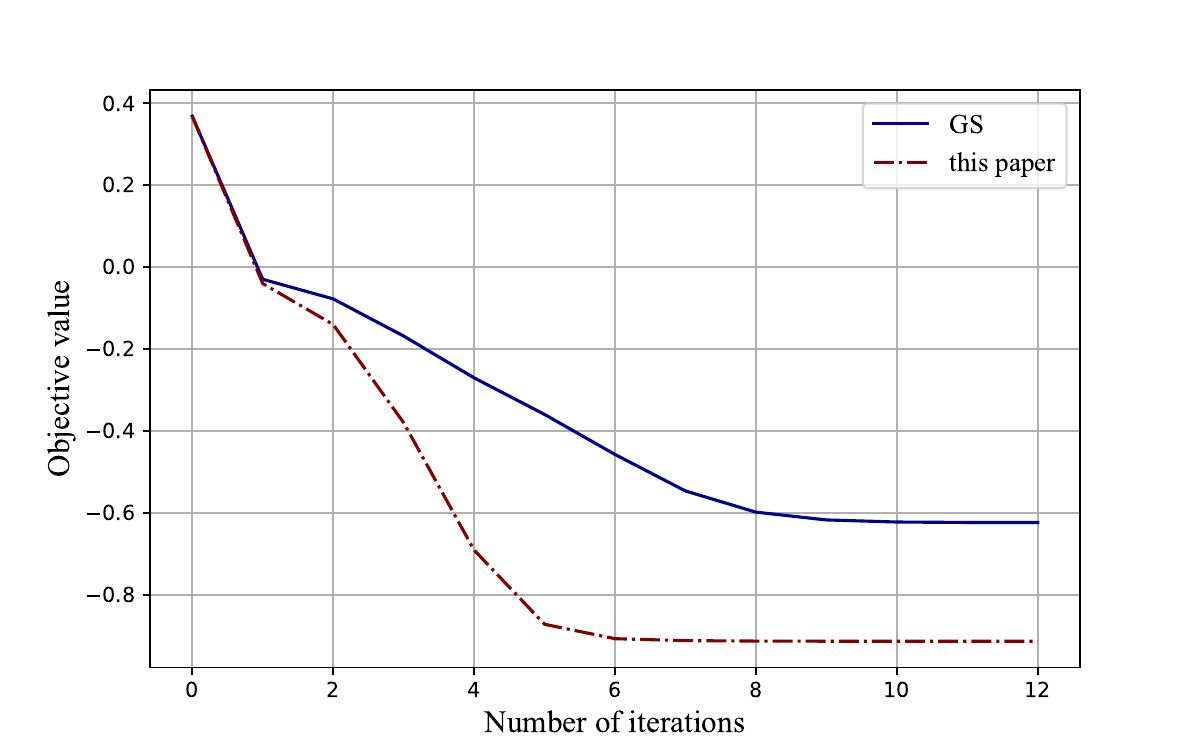}
    \caption{Iterative behavior of two algorithms with $\beta=5$ (1).
    The vertical axis shows the value of the objective function.
    The horizontal axis shows the number of iterations. The initial point is given by 
    \eqref{BVT} with $k=20$.}
    \label{fig:3}
\end{figure}

\begin{figure}[htbp]
    \centering
    \includegraphics[width = \columnwidth]{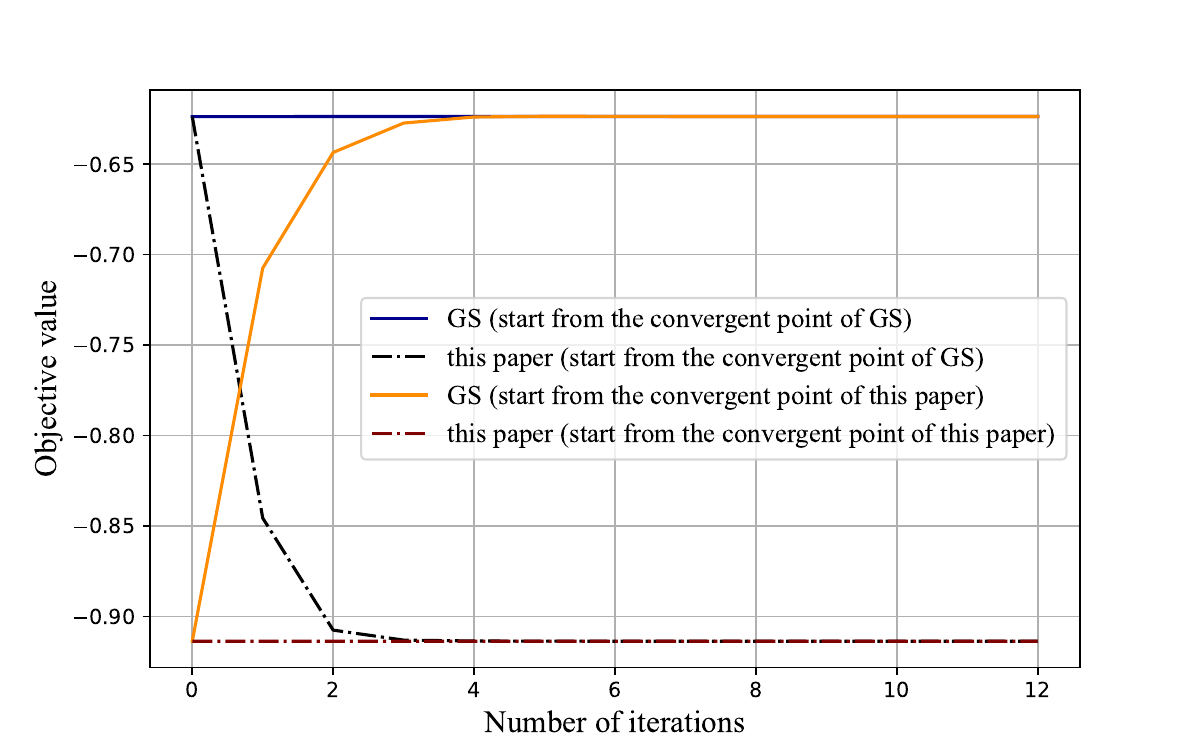}
    \caption{Iterative behavior of two algorithms with $\beta=5$ (2).
    The vertical and horizontal axes show the same as Fig. \ref{fig:3}. 
    Orange solid and brown chain lines express the algorithm \cite{Grimsmo} and our algorithm, respectively
     when the initial point is set to be the convergent point of our algorithm in 
   Fig. \ref{fig:3}.
Blue solid and black chain lines express the algorithm \cite{Grimsmo} and our algorithm, respectively
     when the initial point is set to be the convergent point of the algorithm \cite{Grimsmo} in Fig. \ref{fig:3}.}
    \label{fig:4}
\end{figure}

In the third experiment, we set $\beta=5$. Fig.~\ref{fig:3} illustrates the iterative behaviors of both algorithms when the initial point is fixed at $k=20$.
For the fourth experiment, we maintain the same $\beta$, but vary the initial point by choosing it as the convergent point identified in Fig.~\ref{fig:3}. Fig.~\ref{fig:4} depicts the iterative behaviors of both algorithms using these two distinct initial points.
Remarkably, our algorithm consistently decreases the objective function when the initial point corresponds to the convergent point of their algorithm. Surprisingly, their algorithm exhibits an iterative increase in the objective function when the initial point is the convergent point of our algorithm. It's noteworthy that their algorithm converges to the same point with two different initial points, specifically, \eqref{Eq: different_initial_states} with $k=20$ and the convergent point of our algorithm.
%This fact shows that their convergent point is a saddle point.
Remember that our algorithm is always monotonically decreasing 
the objective function as guaranteed by Theorem \ref{NLT} with the choice $\gamma=\alpha$. However, their algorithm is designed such that the convergent point is a stationary point, i.e., a point where the derivative is zero. Consequently, there are scenarios where their algorithm iteratively increases the objective function, as illustrated in Fig.~\ref{fig:4}.

\section*{Discussion}\Label{S7}
We have generalized the algorithm by \cite{RISB}
by using the concept of a mixture family, which is a key concept of information 
geometry.
Although their algorithm works only with a function defined on the set of all density matrices,
our algorithm works with a function defined on a subset with linear constraints.
For this generalization, 
each iterative step defines an improved element in the subset.
To choose a suitable element of the subset, we have employed the projection to 
the subset by using the idea of information geometry.
Our algorithm is an iterative algorithm,
and each step improves the value of the objective function.
Also, we have shown the convergence to the global minimum
when the objective function satisfies certain conditions.

Since our algorithm has a wider applicability,
our algorithm can be applied to an optimization problem that cannot be solved 
by the algorithm \cite{RISB}.
For example, we can expect that 
it can be applied to the quantum versions of 
the problems explained in the reference \cite{Hmixture} including
the quantum version of em-algorithm \cite{H23}, which covers 
rate distortion theory.
Further, as another example, we have applied our algorithm to 
quantum information bottleneck.
The previous paper \cite{HY} discussed
an algorithm for quantum information bottleneck.
But, it assumes that the system $X$ is a classical system.
Our algorithm can be applied even when the system $X$ is a quantum system.
Also, the paper \cite{HY} proposed an algorithm for 
quantum information bottleneck, 
which works even with a quantum system $X$.
Our algorithm is a generalization of the algorithm \cite{HY}.
When the system $X$ is a classical system,
the previous paper \cite{HY} showed that 
the algorithm by \cite{HY} is the same as the algorithm by \cite{Grimsmo}.

Therefore, it is an important issue whether our algorithm is the same as the algorithm by 
\cite{Grimsmo}.
To answer this problem, we have compared our algorithm with 
the algorithm \cite{Grimsmo}.
As a result, we have found that 
our algorithm has a different behavior from their algorithm,
and has better performances than their algorithm.
Although we have shown that our algorithm monotonically 
decreases the value of the objective function,
our numerical analysis has revealed that
there exists a case that their algorithm increases the value of the objective function.
This difference shows a big advantage of our algorithm over their algorithm.

Our numerical analysis shows that our algorithm does not have 
the convergence to the global minimum as well as the algorithm by 
\cite{Grimsmo}.
Since the necessarily condition \eqref{BK2+}
for the convergence to the global minimum
is rewritten as
the relation
\begin{align}
&\alpha D( \sigma_{TX}\|\sigma_{TX}')
+(\beta-1)  D(\sigma_{T}\|\sigma_{T}')
\nonumber\\
&- \beta D ({\cal R}( \sigma_{TX}) \|{\cal R}( \sigma_{TX}') )\ge 0\Label{XMQ}
\end{align}
for $\sigma_{TX},\sigma_{TX}'\in {\cal M}(\rho_X)$,
our numerical analysis shows that the relation \eqref{XMQ} does not hold.

Our general algorithm covers various functions defined over a set of density matrices 
with linear constraints.
Since it has wider applicability, it is an interesting future topic to find 
its other applications.
Although we have pointed out the above problem for the algorithm by \cite{Grimsmo},
we could not sufficiently clarify the iterative behavior of their algorithm.
It is another interesting future study to analyze its iterative behavior.

Recently, the paper \cite{HSF1} showed that
the algorithm by \cite{RISB} is equivalent to 
the mirror descent algorithm \cite[Section 4.2]{Bubeck}, \cite[Algorithm 1]{HSF1}
when the function
$\Omega$ is a constant times of the gradient of the objective.
Therefore, it is an interesting problem whether
a similar equivalence holds for Algorithm \ref{AL1} under a similar condition.
In addition, 
although the mirror descent algorithm requires us to solve a minimization 
in each iteration,
our algorithm exactly gives the solution of 
the minimization in each iteration, which is an advantage of our algorithm.

\section*{Methods}
\subsection*{Analysis for our general algorithm and Proof of Theorem \ref{NLT}}
Indeed, Algorithm \ref{AL1} is characterized as the iterative minimization of 
the following two-variable function, i.e., the extended objective function;
\begin{align}
J_\gamma(\rho,\sigma):=\gamma D(\rho\|\sigma)+\Tr \rho \Omega[\sigma].
\end{align}
To see this fact, %consider an iterative algorithm, 
we define
\begin{align}
{\cal F}_1[\rho]  := \argmin_{\sigma \in {\cal M}_a}  J _\gamma(\rho,\sigma) ,\quad
{\cal F}_2[\sigma]  := \argmin_{\rho \in {\cal M}_a}  J _\gamma(\rho,\sigma) .
\end{align}
Then, as a generalization of a part of \cite[Lemma 3.2]{RISB}, ${\cal F}_2[\sigma]$ is calculated as follows.
\begin{lemma}\Label{L1}
Then, we have ${\cal F}_2[\sigma] =\Gamma^{(e)}_{{\cal M}_a}[{\cal F}_3[\sigma]] $, i.e., 
\begin{align}
\min_{\rho \in {\cal M}_a}  J_\gamma(\rho,\sigma)&=
J _\gamma(\Gamma^{(e)}_{{\cal M}_a}[{\cal F}_3[\sigma]],\sigma) \nonumber \\
&=
\gamma D(\Gamma^{(e)}_{{\cal M}_a}[{\cal F}_3[\sigma]]\|{\cal F}_3[\sigma])
- \gamma \log \kappa[\sigma]  ,\Label{XMY} \\
 J _\gamma(\rho,\sigma)
 &=\min_{\rho' \in {\cal M}_a}  J _\gamma(\rho',\sigma)
+\gamma D(\rho\| \Gamma^{(e)}_{{\cal M}_a}[{\cal F}_3[\sigma]]) \Label{XMY2UU}  \\
&=J _\gamma(\Gamma^{(e)}_{{\cal M}_a}[{\cal F}_3[\sigma]],\sigma)
+\gamma D(\rho\| \Gamma^{(e)}_{{\cal M}_a}[{\cal F}_3[\sigma]]).
\Label{XMY2} 
\end{align}
\end{lemma}
\begin{proof}
Since $\log {\cal F}_3[\sigma]=- \log \kappa[\sigma]
+\log\sigma - \frac{1}{\gamma} \Omega[\sigma]$, we have
\begin{align}
&J _\gamma(\rho,\sigma)
=\gamma \Tr \rho (\log \rho- \log \sigma + \frac{1}{\gamma} \Omega[\sigma]) \nonumber\\
=&\gamma \Tr \rho (\log \rho- \log {\cal F}_3[\sigma]- \log \kappa[\sigma]) \nonumber\\
=&\gamma D(\rho\| {\cal F}_3[\sigma])-\gamma \log \kappa[\sigma] \nonumber\\
=&\gamma D(\rho\| \Gamma^{(e)}_{{\cal M}_a}[{\cal F}_3[\sigma]])
+\gamma D(\Gamma^{(e)}_{{\cal M}_a}[{\cal F}_3[\sigma]]\|{\cal F}_3[\sigma]) \nonumber\\
 &- \gamma\log \kappa[\sigma]  .\Label{ASS4}
\end{align}
Then, the minimum is given as \eqref{XMY}, and it is realized with $\Gamma^{(e)}_{{\cal M}_a}[{\cal F}_3[\sigma]]$.

Applying \eqref{XMY} into the final line of \eqref{ASS4},
we obtain \eqref{XMY2UU}.
Since the minimum in \eqref{XMY2UU} is realized when 
$\rho'=\Gamma^{(e)}_{{\cal M}_a}[{\cal F}_3[\sigma]]$, 
we obtain \eqref{XMY2}.
\end{proof}

As a generalization of another part of \cite[Lemma 3.2]{RISB}, we can calculate ${\cal F}_1[\sigma]$ as follows.
\begin{lemma}\Label{L2}
Assume that two density matrices $\rho,\sigma \in {\cal M}_a$ satisfy the 
condition \eqref{BK1+}.
Then, we have ${\cal F}_1[\rho] =\rho$, i.e., 
\begin{align}
J _\gamma(\rho,\sigma)\ge J _\gamma(\rho,\rho) .
\end{align}
\end{lemma}
\begin{proof}
Eq. \eqref{BK1+} guarantees that
\begin{align}
&J _\gamma(\rho,\sigma)-J _\gamma(\rho,\rho) \nonumber\\
=&
\gamma D(\rho\|\sigma)-\Tr \rho 
(\Omega[\rho]- \Omega[\sigma])
\ge 0 \Label{XMY5}.
\end{align}
\end{proof}

Now, we prove Theorem \ref{NLT}.
Due to Lemma \ref{L2}, 
when all pairs $(\rho^{(t+1)},\rho^{(t)})$ satisfy \eqref{BK1+}, 
the relations 
\begin{align}
{\cal G}(\rho^{(t)}) &=J _\gamma(\rho^{(t)},\rho^{(t)})\ge
J _\gamma(\rho^{(t+1)},\rho^{(t)}) \nonumber\\
 & \ge J _\gamma(\rho^{(t+1)},\rho^{(t+1)})= {\cal G}(\rho^{(t+1)})
\Label{SAC}
\end{align}
hold under Algorithm \ref{AL1}. % and the condition (A1).
The relation \eqref{SAC} guarantees that Algorithm \ref{AL1} converges to 
a local minimum.
As a generalization of \cite[Theorem 3.3]{RISB},  
the following theorem discusses the convergence to the global minimum 
and the convergence speed.

%\begin{proof}
Since $\{{\cal G}(\rho^{(t)}) \}$
is monotonically decreasing for $t$,
we have 
\begin{align}
\lim_{n \to \infty}
{\cal G}(\rho^{(t)})-{\cal G}(\rho^{(t+1)})
  =0\Label{AAS}.
\end{align}
Using \eqref{XMY2}, we have
  \begin{align}
&{\cal G}(\rho^{(t)})
=J_{\gamma}(\rho^{(t)},\rho^{(t)})  \nonumber \\
=& \gamma 
 D(\rho^{(t)}\| \rho^{(t+1)})
+
J_{\gamma}( \rho^{(t+1)}, \rho^{(t)}) \nonumber \\
\ge 
& \gamma 
 D(\rho^{(t)}\| \rho^{(t+1)})
+
{\cal G}( \rho^{(t+1)}) .
\end{align}
Thus, we have
\begin{align}
 \gamma  D(\rho^{(t)}\| \rho^{(t+1)})
\le 
{\cal G}(\rho^{(t)})-{\cal G}(\rho^{(t+1)})
  \Label{AAS3}.
\end{align}
Since
due to \eqref{AAS} and \eqref{AAS3},
the sequence $\{{\cal G}(\rho^{(t)})\}$ is a Cauchy sequence, it converges. 
Then, we obtain Theorem \ref{NLT}.
%\end{proof}

Since Theorem \ref{NLT} guarantees the convergence of $\{\rho^{(t)}\}_t$,
we denote the convergent with the initial point $\rho$ by $C(\rho)$.
When the condition \eqref{BK1+} holds, and
a state $\rho'$ satisfies the condition \eqref{BK2+}, 
Theorem \ref{NLT} guarantees 
\begin{align}
{\cal G}(C(\rho)) \le
{\cal G}(\rho') .
\end{align}

% \appendices
%\onecolumn\newpage
%\appendix
%\begin{widetext}
\subsection*{Preparation for proof of Theorem \ref{TH1}}
%\subsubsection{}
To show Theorem \ref{TH1}, we prepare the following lemma.
\begin{lemma}\Label{LLX}
For any density matrices $\rho,\sigma \in {\cal M}_a$, we have
\begin{align}
&D(\rho\| \sigma )- D(\rho\| \Gamma^{(e)}_{{\cal M}_a}[{\cal F}_3[\sigma]])  \nonumber\\
=&\frac{1}{\gamma}J_\gamma(\Gamma^{(e)}_{{\cal M}_a}[{\cal F}_3[\sigma]],\sigma) 
-\frac{1}{\gamma}{\cal G}(\rho)
+\frac{1}{\gamma}\Tr \rho  
\Big(\Omega[\rho]-\Omega[\sigma]\Big)
\Label{XM1} \\
=&\frac{1}{\gamma}{\cal G}(\Gamma^{(e)}_{{\cal M}_a}[{\cal F}_3[\sigma]]) -\frac{1}{\gamma}{\cal G}(\rho) 
+
D(\Gamma^{(e)}_{{\cal M}_a}[{\cal F}_3[\sigma]]\|\sigma)
\nonumber \\
&-\frac{1}{\gamma} \Tr \Gamma^{(e)}_{{\cal M}_a}[{\cal F}_3[\sigma]] 
( \Omega[\Gamma^{(e)}_{{\cal M}_a}[{\cal F}_3[\sigma]]]-\Omega[\sigma])\nonumber\\
&+\frac{1}{\gamma}\Tr \rho  
\Big(\Omega[\rho]-\Omega[\sigma]\Big)
.\Label{XM2}
\end{align}
%In addition, when $\Omega$ is defined for any distribution in ${\cal P}({\cal X})$,
%the above relations holds for any distribution $\sigma \in {\cal P}({\cal X})$.
\end{lemma}

\begin{proof}

We have
\begin{align}
-\Tr \rho  \Omega[\sigma]
=-{\cal G}(\rho)
+\Tr \rho  
\Big(\Omega[\rho]-\Omega[\sigma]\Big).
\Label{MLP}
\end{align}
Using \eqref{MLP}, 
we have
\begin{align}
&D(\rho\| \sigma)- D(\rho\| \Gamma^{(e)}_{{\cal M}_a}[{\cal F}_3[\sigma]]
)  \nonumber\\
=&\Tr \rho (\log \Gamma^{(e)}_{{\cal M}_a}[{\cal F}_3[\sigma]]- \log \sigma)\nonumber \\
=&\Tr \rho \Big(
\log \Gamma^{(e)}_{{\cal M}_a}[{\cal F}_3[\sigma]]
-\log {\cal F}_3[\sigma]
+\log {\cal F}_3[\sigma]
- \log \sigma \Big) \nonumber\\
\stackrel{(a)}{=}&
D(\rho\| {\cal F}_3[\sigma])
-D(\rho\| \Gamma^{(e)}_{{\cal M}_a}[{\cal F}_3[\sigma]])  \nonumber\\
&+\Tr \rho \Big(
-\frac{1}{\gamma} \Omega[\sigma] -\log \kappa[\sigma] \Big)\nonumber \\
\stackrel{(b)}{=}&
D(\Gamma^{(e)}_{{\cal M}_a}[{\cal F}_3[\sigma]]\| {\cal F}_3[\sigma])
-\log \kappa[\sigma] 
-\frac{1}{\gamma}\Tr \rho  \Omega[\sigma]\nonumber\\
\stackrel{(c)}{=}&
\frac{1}{\gamma}J_\gamma(\Gamma^{(e)}_{{\cal M}_a}[{\cal F}_3[\sigma]],\sigma) 
-\frac{1}{\gamma}{\cal G}(\rho)
+\frac{1}{\gamma}\Tr \rho  
\Big(\Omega[\rho]-\Omega[\sigma]\Big)\nonumber\\
\stackrel{(d)}{=}&
\frac{1}{\gamma}{\cal G}(\Gamma^{(e)}_{{\cal M}_a}[{\cal F}_3[\sigma]]) -\frac{1}{\gamma}{\cal G}(\rho)
+D(\Gamma^{(e)}_{{\cal M}_a}[{\cal F}_3[\sigma]]\|\sigma)
\nonumber\\
&-\frac{1}{\gamma}\Tr \Gamma^{(e)}_{{\cal M}_a}[{\cal F}_3[\sigma]] 
( \Omega[\Gamma^{(e)}_{{\cal M}_a}[{\cal F}_3[\sigma]]]-\Omega[\sigma])\nonumber\\
&+\frac{1}{\gamma}\Tr \rho  
\Big(\Omega[\rho]-\Omega[\sigma]\Big),
\end{align}
where each step is shown as follows.
$(a)$ follows from the definition of ${\cal F}_3(\sigma) $. 
$(c)$ follows from \eqref{XMY} and \eqref{MLP}. 
$(d)$ follows from \eqref{XMY5}. 
$(b)$ follows from the relations 
\begin{align*}
D(\rho\| {\cal F}_3[\sigma])&=D(\rho\| \Gamma^{(e)}_{{\cal M}_a}[{\cal F}_3[\sigma]])
+D(\Gamma^{(e)}_{{\cal M}_a}[{\cal F}_3[\sigma]]\| {\cal F}_3[\sigma]) ,
\end{align*}
which are shown by Phythagorean 
equation \cite{Amari-Nagaoka}, \cite[Proposition 1]{H23}.
\end{proof}

\subsection*{Proof of Theorem \ref{TH1}}
\noindent{\bf Step 1:}\quad
The aim of this step is showing the following inequality;
\begin{align}
D(\rho_{*,a}\| \rho^{(t)})- D(\rho_{*,a}\| \rho^{(t+1)}) \ge \frac{1}{\gamma}{\cal G}(\rho^{(t+1)}) -\frac{1}{\gamma}{\cal G}(\rho_{*,a})\Label{XPZ2}
\end{align}
for $t=1, \ldots, t_0-1$.
We show these relations by induction.

For any $t$, 
by using the relation $ \Gamma^{(e)}_{{\cal M}_a}[{\cal F}_3[\rho^{(t)}]]=\rho^{(t+1)}$, 
the application of \eqref{XM2} of Lemma \ref{LLX} to the case with 
$\sigma=\rho^{(t)}$ and $\rho=
\rho_{*,a}$
yields
\begin{align}
&D(\rho_{*,a}\| \rho^{(t)})- D(\rho_{*,a}\| \rho^{(t+1)})  \nonumber\\
=&\frac{1}{\gamma}{\cal G}( \rho^{(t+1)}) 
-\frac{1}{\gamma}{\cal G}(\rho_{*,a}) %\nonumber\\
+D(\Gamma^{(e)}_{{\cal M}_a}[{\cal F}_3[\rho^{(t)}]]\|\rho^{(t)})\nonumber\\
&-\frac{1}{\gamma}\Tr  \Gamma^{(e)}_{{\cal M}_a}[{\cal F}_3[\rho^{(t)}]] 
( \Omega[\Gamma^{(e)}_{{\cal M}_a}[{\cal F}_3[\rho^{(t)}]]-\Omega[\rho^{(t)}])
\nonumber\\
&+\frac{1}{\gamma}\Tr  \rho_{*,a}  
\Big(\Omega[\rho_{*,a}]-\Omega[\rho^{(t)}]\Big) \\
=&\frac{1}{\gamma}{\cal G}( \rho^{(t+1)}) 
-\frac{1}{\gamma}{\cal G}(\rho_{*,a}) 
+D(\rho^{(t+1)}\|\rho^{(t)})
\nonumber\\
&-\frac{1}{\gamma}\Tr  \rho^{(t+1)} 
( \Omega[\rho^{(t+1)}]-\Omega[\rho^{(t)}])
\nonumber\\
&+\frac{1}{\gamma}\Tr  \rho_{*,a}  
\Big(\Omega[\rho_{*,a}]-\Omega[\rho^{(t)}]\Big)
%-D(\Gamma^{(e)}_{{\cal M}_a}[{\cal F}_3[\rho^{(t)}]]\| \rho^{(t+1)} ) 
.\Label{XME2}
\end{align}
Since two densities $\Gamma^{(e)}_{{\cal M}_a}[{\cal F}_3[\rho^{(t)}]]$ and 
$\rho^{(t)}$ satisfy the conditions \eqref{BK1+} and \eqref{BK2+},
we have
\begin{align}
\hbox{\rm(RHS of \eqref{XME2})} 
\ge %\stackrel{(a)}{\ge} &
\frac{1}{\gamma}{\cal G}(\rho^{(t+1)}) -\frac{1}{\gamma}{\cal G}(\rho_{*,a}).\Label{AXE}
\end{align}
The combination of \eqref{XME2} and \eqref{AXE} implies \eqref{XPZ2}.

\noindent{\bf Step 2:}\quad
The aim of this step is to show \eqref{XME}.
Lemmas \ref{L1} and \ref{L2}
guarantee that
\begin{align}
{\cal G}(\rho^{(t+1)}) \le {\cal G}(\rho^{(t)}) \Label{AMK}.
\end{align}
We have
\begin{align}
& \frac{t_0}{\gamma} \Big(%\min_{t=2,\ldots,k+1}{\cal G}(\rho^{(t)})
{\cal G}(\rho^{(t_0+1)}) - {\cal G}(\rho_{*,a}) \Big)
%\nonumber\\
\stackrel{(a)}{\le}  \frac{1}{\gamma}\sum_{t=1}^{t_0}
{\cal G}(\rho^{(t+1)}) -{\cal G}(\rho_{*,a}) %J^{(t)}-J^*
\nonumber\\
\stackrel{(b)}{\le}& \sum_{t=1}^{t_0}
D(\rho_{*,a}\| \rho^{(t)})- D(\rho_{*,a}\| \rho^{(t+1)})  %\nonumber\\
\nonumber\\
=& 
D(\rho_{*,a}\| \rho^{(1)})-
D(\rho_{*,a}\| \rho^{(t_0+1)})  
\le D(\rho_{*,a}\| \rho^{(1)}),
\end{align}
where $(a)$ and $(b)$ follow from \eqref{AMK} and \eqref{XPZ2}, respectively.

\begin{remark}
When the condition \eqref{BK2+} does not hold,
the above proof does not work.
However, when 
$D(\rho^{(t+1)}\|\rho^{(t)})
-\frac{1}{\gamma}\Tr  \rho^{(t+1)} 
( \Omega[\rho^{(t+1)}]-\Omega[\rho^{(t)}])
+\frac{1}{\gamma}\Tr  \rho_{*,a}  
\Big(\Omega[\rho_{*,a}]-\Omega[\rho^{(t)}]\Big)\ge 0$,
the above proof does work.
Maybe, there is a possibility 
that this proof works locally with sufficiently large $\gamma$.
\end{remark}

\subsection*{Derivation of \eqref{NMD}}
For a state $\sigma_{TX}$,
we have
\gl{
\begin{align}
&\Tr\sigma_{TX} {\Omega}_{\alpha,\beta}[\sigma_{TX}] \nonumber \\
=&
\Tr\sigma_{TX} 
\Big(\alpha \log \sigma_{TX}
-\alpha \log \rho_{X}
+(\beta-1) \log \sigma_{T}
\nonumber\\
& +\beta 
 {\cal R}^* ( \log {\cal R}(\rho_{X}) 
- \log {\cal R}( \sigma_{TX}) ) \Big) \nonumber \\
=&
\alpha \Tr\sigma_{TX} 
\Big( \log \sigma_{TX}
- \log \rho_{X}
- \log \rho_{T}\Big)
\nonumber\\
& -(1-\alpha)\Tr\sigma_{TX} \log \sigma_{T} \nonumber \\
&+\beta \Tr\sigma_{TX}\Big(
 \log \sigma_{T} 
+
 {\cal R}^* ( \log {\cal R}(\rho_{X}) 
- \log {\cal R}( \sigma_{TX}) ) \Big) \nonumber \\
=&\alpha I(T;X)[\sigma_{TX}]+(1-\alpha)H(T)[\sigma_{TX}]
\nonumber\\
&-\beta I(T;Y)[{\cal R}(\sigma_{TX})]
\end{align}}
% \begin{align}
% &\Tr\sigma_{TX} {\Omega}_{\alpha,\beta}[\sigma_{TX}] \nonumber \\
% =&
% \Tr\sigma_{TX} 
% \Big(\alpha \log \sigma_{TX}
% -\alpha \log \rho_{X}
% +(\beta-1) \log \sigma_{T}
% +\beta 
%  {\cal R}^* ( \log {\cal R}(\rho_{X}) 
% - \log {\cal R}( \sigma_{TX}) ) \Big) \nonumber \\
% =&
% \alpha \Tr\sigma_{TX} 
% \Big( \log \sigma_{TX}
% - \log \rho_{X}
% - \log \rho_{T}\Big)
% -(1-\alpha)\Tr\sigma_{TX} \log \sigma_{T} \nonumber \\
% &+\beta \Tr\sigma_{TX}\Big(
%  \log \sigma_{T} 
% +
%  {\cal R}^* ( \log {\cal R}(\rho_{X}) 
% - \log {\cal R}( \sigma_{TX}) ) \Big) \nonumber \\
% =&\alpha I(T;X)[\sigma_{TX}]+(1-\alpha)H(T)[\sigma_{TX}]
% -\beta I(T;Y)[{\cal R}(\sigma_{TX})]
% \end{align}
because $\Tr\sigma_{TX} \log \sigma_{T}=\Tr\sigma_{T} \log \sigma_{T}$.
Hence, we obtain \eqref{NMD}.

\subsection*{Derivation of \eqref{MMS}}
First, we have
\begin{align}
D( {\cal R}(\sigma_{TX} )\|{\cal R}(\sigma_{TX}'))
\ge 
D( P_{T} \|Q_{T}).
\end{align}
Then, we have
\gl{\begin{align}
& D_{\Omega_{\alpha,\beta}}( \sigma_{TX} \|\sigma_{TX}') \nonumber \\
=& (\beta-1)D( \sigma_{T} \|\sigma_{T}')
+\alpha D( \sigma_{TX}\|\sigma_{TX}')
\nonumber\\
&-\beta D( {\cal R}(\sigma_{TX} )\|{\cal R}(\sigma_{TX}'))
\nonumber \\
\le & -D( \sigma_{T} \|\sigma_{T}')
+\alpha D( \sigma_{TX}\|\sigma_{TX}') 
\le 
\alpha D( \sigma_{TX}\|\sigma_{TX}').
\end{align}}
Hence, we have \eqref{MMS}.

\bigskip
\noindent\textbf{DATA AVAILABILITY}\\	
The data that support the findings of this study are
available from the corresponding author upon request.

\bigskip

\noindent\textbf{CODE AVAILABILITY}\\	
The codes used for numerical analysis are available from
the corresponding author upon request.

\bigskip
\noindent\textbf{ACKNOWLEDGEMENTS}\\	
%\section*{Acknowledgement}
MH was supported in part by the National Natural Science Foundation of China under Grant 62171212.

\bigskip
\noindent\textbf{AUTHOR CONTRIBUTIONS}\\	
MH proposed this project. MH
derived the theoretical results. 
MH and GL designed the numerical experiments.
GL performed numerical experiments and plotted the figures. 
MH and GL wrote the paper.

\bigskip
\noindent\textbf{COMPETING INTERESTS}\\	
The authors declare no competing interests.


\begin{thebibliography}{99}
\bibitem{Nagaoka} 
H. Nagaoka,
``Algorithms of Arimoto-Blahut type for computing quantum channel capacity,''
{\em Proc. 1998 IEEE Int. Symp. Information Theory (ISIT 1998)}, 
Cambridge, MA, USA, 16-21 Aug. 1998, pp. 354.

\bibitem{Blahut} 
R. Blahut, 
``Computation of channel capacity and rate-distortion functions,''
{\em IEEE Trans. Inform. Theory}, 
vol. 18, no. 4, 460 -- 473 (1972).

\bibitem{Arimoto} 
S. Arimoto, ``An algorithm for computing the capacity of arbitrary discrete memoryless channels,'' 
{\em IEEE Trans. Inform. Theory}, vol. 18, no. 1, 14 -- 20 (1972).

\bibitem{RISB}
N. Ramakrishnan, R. Iten. V. B. Scholz, and M. Berta,
``Computing Quantum Channel Capacities,''
{\em IEEE Trans. Inform. Theory}, 
vol. 67, 946 -- 960 (2021).

\bibitem{Hmixture}
M. Hayashi, 
``Iterative minimization algorithm on mixture family,''
arXiv:2302.06905  (2023)

\bibitem{Tishby}
N. Tishby, F. C. Pereira, and W. Bialek.
``The information bottleneck method,''
In The 37th annual Allerton Conference on Com-
munication, Control, and Computing, pages
368--377. Univ. Illinois Press, 1999. 
%DOI: 10.48550/arXiv.physics/0004057.

\bibitem{Grimsmo}
A. L. Grimsmo and S. Still. 
``Quantum predictive filtering,''
{\em Phys. Rev. A}, 94, 012338,
(2016). 

\bibitem{BPP}
L. Banchi, J. Pereira, and S. Pirandola. 
``Generalization in quantum machine learning: A quantum information standpoint,''
{\em PRX Quantum}, 2, 040321 (2021).

\bibitem{DHW}
N. Datta, C. Hirche, and A. Winter, 
``Convexity and operational interpretation of the quantum information bottleneck function,'' 
In {\em 2019 IEEE International Symposium
on Information Theory (ISIT)}, pages 1157--1161 (2019).

\bibitem{HW}
C. Hirche and A. Winter, 
``An alphabet-size bound for the information bottleneck function,''
In {\em 2020 IEEE International Symposium on Information Theory (ISIT)}, 
pages 2383--2388 (2020).

\bibitem{SCKW}
S. Salek, D. Cadamuro, P. Kammerlander, and K. Wiesner. 
``Quantum rate-distortion coding of relevant information,''
{\em IEEE Transactions on Information Theory}, 65(4)
2603--2613 (2019). 

\bibitem{HY}
M. Hayashi and Y. Yang,
``Efficient algorithms for quantum information bottleneck,''
{\em Quantum}, {\bf 7}, 936 (2023).

\bibitem{HSF1}
K. He, J. Saunderson, and H. Fawzi,
``A Bregman Proximal Perspective on Classical and Quantum Blahut-Arimoto Algorithms,''
arXiv:2306.04492 (2023).

\bibitem{HSF2}
K. He, J. Saunderson, and H. Fawzi,
``Efficient Computation of the Quantum Rate-Distortion Function,''
arXiv:2309.15919 (2023).

\bibitem{OHS}
O. Faust, H. Fawzi, and J. Saunderson,
``A Bregman Divergence View on the Difference-of-Convex Algorithm,''
{\em Proceedings of Machine Learning Research},
vol. 206, pp. 3427 -- 3439.
{\em 26th International Conference on Artificial Intelligence and Statistics, AISTATS 2023}, Valencia, April 25 -- 27, (2023).

\bibitem{Amari-Nagaoka} 
S. Amari and H. Nagaoka, 
{\em Methods~of~Information~Geometry} (AMS and Oxford, 2000).

\bibitem{H23}
M. Hayashi, 
``Bregman divergence based em algorithm and its application to classical and quantum rate distortion theory,''
{\em IEEE Trans. Inform. Theory}, 
Volume: 69, Issue: 6, 3460 -- 3492 (2023).

\bibitem{Strouse}
D. J. Strouse and David J. Schwab. 
``The Deterministic Information Bottleneck,''
{\em Neural Computation}, vol. 29, no. 6, pp. 1611-1630, (2017). 

\bibitem{Bubeck}
S. Bubeck,
``Convex optimization: Algorithms and complexity,''
{\em Foundations and Trends in Machine Learning}, 
8(3-4), 231-357 (2015). 
 
\end{thebibliography}
\end{document}